\documentclass{LMCS}

\def\dOi{10(3:20)2014}
\lmcsheading%
{\dOi}
{1--24}
{}
{}
{May\phantom.~13, 2013}
{Sep.~12, 2014}
{}

\ACMCCS{[{\bf Mathematics of computing}]: Probability and statistics;
  [{\bf Theory of computation}]: Models of computation}


\usepackage{amsmath}
\usepackage{amssymb}
\usepackage{amsmath}
\usepackage{amsfonts}
\usepackage{amssymb}
\usepackage[all,cmtip]{xy}
\usepackage{hyperref}

\theoremstyle{plain}

\begin{document}

\title[Fourier spectra of measures and Kolmogorov complexity]{Fourier spectra of  measures associated with algorithmically random Brownian motion}

\author[W.~L.~Fouch\'e]{Willem L. Fouch\'e}	
\author[S.~Mukeru]{Safari Mukeru}
\author[G.~Davie]{George Davie}

\address{Department of Decision Sciences, School of Economic Sciences, University of South Africa, P.O. Box 392, Pretoria, 0003, South Africa.}	
\email{\{fouchwl,mukers,davieg\}@unisa.ac.za}  


\maketitle
\begin{center}
{\it To Dieter Spreen on the occasion of his 65th birthday.} 
\end{center}
\begin{abstract}
\noindent In this paper we study the behaviour at infinity of the Fourier transform of Radon measures supported by the images of fractal sets under an algorithmically random Brownian motion. We show that, under some computability conditions on these sets, the Fourier transform of the associated measures have, relative to the Hausdorff dimensions of these sets, optimal
asymptotic decay at infinity. The argument relies heavily on a direct characterisation, due to Asarin and Pokrovskii, of algorithmically random  Brownian motion in terms of the prefix-free Kolmogorov complexity of finite binary sequences. The study also necessitates a closer look at the potential theory over fractals from a
computable point of view. 
\end{abstract}

\section{Introduction}

\begin{quote}
\emph {It is important to add here that order is not to be identified
with predictability. Predictability is a property of a special kind of
order such that a few steps determine the whole order (...as in
curves of low degree) – but there can be complex and subtle
orders which are not in essence related to predictability \linebreak (...a
good painting is highly ordered, and yet this order does not permit one part to be predicted from another).} (David Bohm \cite{bohm:1} p 149.)

\end{quote} 
\vspace{5mm}
In 1869 Heine proposed to Cantor the problem of determining whether or not every trigonometric series $\sum_{n\in \mathbb{Z}}c(n) e^{i n x}$ that converges to $0$ at all real numbers $x$ will have all its coefficients $c(n)$ necessarily equal to $0$ or equivalently, whether two trigonometric series which converge to
the same limit for every real number $x$ are (formally) equal in the sense that their coefficients are the same. In 1870, by making extensive use of Riemann's work in his Habilitationsschrift (1854) on trigonometric series, Cantor proved that the answer to Heine's problem was in the affirmative.
Eventually Cantor showed that the same holds true if a trigonometric series converges everywhere with the possible exception of a closed countable set. This is arguably the first example of a mathematical argument that is based on a transfinite induction. For more background on this historical development, the reader is referred to the book \cite{kelo:1}.

These works of Riemann and Cantor pioneered a very highly developed theory of so-called sets of uniqueness in harmonic analysis with wide-ranging implications for number theory (Diophantine approximation) and descriptive set theory. For the former, the reader could consult the book by Meyer \cite{Mey:01} and for the
latter the book by Kechris and Louveau \cite{kelo:1}.

A set $E \subset [0, 1]$ is a {\it set of uniqueness} if every trigonometric series which converges to $0$ for $x\in [0, 1]\setminus E$ is identically $0$, or equivalently, if any two trigonometric series that converge to the same limit for every real $ x\in [0, 1]\setminus E $ are identical. Intuitively, this means that
$[0, 1]\setminus E $ is sufficiently  ``large'' to ensure that if the two series already agree pointwise on it, then they are formally the same , in the sense that they have the same coefficients.
A set which is not a set of uniqueness is called a {\it set of multiplicity}. More specifically, a subset $M \subset [0, 1]$ is a set of multiplicity if there exist two (formally) distinct trigonometric series that converge to the same limit outside $M$.

It is well-known that if $E$ is a set of uniqueness and is Lebesgue-measurable, then its Lebesgue measure is $0$. It was first conjectured by Luzin  that all Lebesgue null sets should be sets of uniqueness until Menshov (1916) \cite{menshov:1} constructed an example of a closed null set of multiplicity. Since then major progress has been made on
this problem.  Salem and Zygmund \cite{SaZy:1} characterized Cantor type sets of fixed ratio $\xi$ as being sets of uniqueness solely in terms of the number theoretical structure of $\xi$.
 (The resulting Cantor set is a set of  multiplicity if and only if $\xi^{-1}$ is a so-called Pisot number. For a more modern proof, the reader is referred to the book \cite{Mey:01}.) 
However, the characterisation of sets of uniqueness is very far from being complete. Indeed, as is discussed in \cite{kelo:1}, it follows from arguments by Solovay (unpublished) and independently by Kaufman \cite{kau:1} that the complexity of the problem is, from the viewpoint of descriptive set theory, of an
intrinsic nature. 

For the purpose of this paper, it is important to note that a compact subset $E$ of the unit interval is a set of multiplicity if and only if there is a distribution $T$ (in the sense of Schwartz) , also called a temperate distribution, supported by $E$ such that its Fourier transform is such that $\hat{T}(u)\rightarrow 0$ as $|u|\rightarrow \infty$. (See, for example Meyer \cite{Mey:01}, Theorem 1, p 81.) In particular, a compact set $E$ will be a set of multiplicity if it is supported by a Radon probability measure $\mu$ with the property that its Fourier transform $\hat{\mu}(u)$ approaches zero as $|u| \rightarrow \infty$.

The following question was posed by Beurling and solved in the affirmative by Salem \cite{Sa:01} in 1950:
\begin{quote}
{\it  Given a number $\gamma \in (0, 1)$, does there exist a closed set on the line whose Hausdorff dimension is $\gamma$ and that carries a non-zero Radon measure $\mu$ whose Fourier transform $\hat \mu (u) = \int_{\mathbf{R}} e^{iux}d\mu(x)$ is dominated by a constant times 
$|u|^{-\gamma/2}$ as $|u| \rightarrow \infty$? }
\end{quote}
Such sets are instances of what are now called {\it Salem sets}. It is well-known  (see for example \cite[pp 162-163]{Ma:01})  that given a compact subset $E$ of $[0, 1]$ with Hausdorff dimension $\gamma \in (0, 1)$, the number $\gamma/2$ is critical for Beurling's question to have an affirmative answer since any Radon  measure $\mu$ supported by $E$ is such that the function 
$u \mapsto \vert u \vert^{\alpha} \hat\mu(u)$ is not bounded for any $\alpha> \gamma/2$. This is a direct consequence of (\ref{equation:energy}). See  also Chapter 2 in \cite{Muk:01} for an extensive discussion.

Note that a set meeting Beurling's condition will necessarily be a set of multiplicity \cite[p 81]{Mey:01} since every Radon measure is a Schwartz distribution. Salem \cite{Sa:01} proved that the answer to Beurling's question is in the affirmative by constructing for every $\gamma$ in the unit
interval, a {\it random} Radon measure $\mu$ (over a convenient probability space) whose support has Hausdorff dimension $\gamma$ and which satisfies his requirement with probability one.  By contrast, we should mention  interesting results on the construction of non-random Salem sets, as can be found in  \cite{bluhm:1}.

It follows from the results of this paper (see Theorem \ref{theorem:effectivebeurling}) that such sets can also be constructed by looking at Cantor type ternary sets $E$ with computable ratios
$\xi$ and then considering the image of $E$ under an algorithmically random Brownian motion.  Along these lines we shall be able to find, for every algorithmically random infinite binary string $\omega$ and  every computable real $\gamma$ in the unit interval,  a $\Pi_2^0(\omega, \gamma)$ compact set $S_\gamma(\omega)$ which is a Salem set of Hausdorff dimension $\gamma$. We emphasise that these sets are uniformly definable in $\gamma$ and $\omega$.

This means that \emph{each}  algorithmically random binary string $\omega$ will answer Beurling's question in the affirmative for \emph{all}  $\gamma$ which are computable. These sets can be uniformly constructed from each specific algorithmically random $\omega$  in a manner which is  $\Pi_2^0 $--definable  in $\omega$ and the dimension $\gamma$.

A Brownian motion on the unit interval is  {\it
algorithmically random} if it meets all effective (Martin-L\" of) statistical tests expressed in terms of
the statistical events associated with Brownian motion on the unit
interval.  The class of functions corresponds exactly, in  the language of Weihrauch \cite{wei:1,wei:2},  G\'acs \cite{gac:1} and specialised by Hoyrup and Rojas \cite{horo:1},  in the context of algorithmic randomness, to the Martin-L\" of random elements of the computable measure space 
\[ \mathcal{R}= (C_0[0,1],d,B,W),\] where $C_0[0,1]$ is the set of the continuous functions on the unit interval that vanish at the origin, $d$ is the metric induced by the uniform norm, $B$ is the countable set of piecewise linear functions vanishing at the origin with slopes and points of non-differentiability all
rational numbers, and where $W$ is the Wiener measure.
 
In this paper we shall show that if $E$ is an ``effectively closed" subset of the unit interval having Hausdorff dimension $\beta$ at most $1/2$, and if $\phi$ is an algorithmically random Brownian motion, then the image $\phi(E)$ is a Salem set of Hausdorff dimension $2 \beta$. Our proof involves a finitisation of Kahane's
arguments in Kahane \cite[p 253]{kah:1} together with a direct Kolmogorov complexity theoretical argument involving the incompressibility coefficient of an algorithmically random binary string \cite{Da:01}. We also rely  on the original characterisation of Asarin and Pokrovskii \cite{ap:1} of algorithmically random Brownian motions $\phi$ in terms of finitary zig-zag paths $x_n$ having high Kolmogorov complexity and which
approximate $\phi$ in a uniform and computable manner. We shall utilize  a Fourier analytic representation (a Franklin-Wiener series) of an  algorithmically random Brownian motion \cite{fo:2}. For the main theorem  we develop a suitable constructivisation of Frostman's potential theoretic approach to the theory of Hausdorff dimension. Another approach to Frostman's lemma in the context of algorithmic randomness can be found in Reimann \cite{Reimann}.

This paper could be viewed as a further contribution towards the understanding of the sample path properties of algorithmically random Brownian motion. For similar results the reader may consult Fouch\'e \cite{fo:1, fo:2, fo:4,fo:5}, Kjos-Hanssen and Nerode \cite{KHNe:1, KHNe:2}, Kjos-Hanssen and Szabados \cite{KHSZ:01}
and Potgieter \cite{Po:01}. More results on the {\it Fourier structure} of sample paths of algorithmically random Brownian motion can be found in \cite{fo:7}.

The paper is organised as follows: In Section 2, we discuss  preliminaries from Brownian motion,  fractal geometry and the theory of Kolmogorov complexity. We also introduce  ideas towards viewing Frostman's potential theory over compact subsets of the unit interval from the viewpoint of computable analysis. The main results are presented in Section 3. In Section 4, we discuss some Fourier analytical properties of a finite binary string with high Kolmogorov
complexity. The proof of the main result is given in Section 5 which is based on  key estimates (a finitisation of Kahane's work \cite{kah:1}) related to the even  moments of the Fourier transform of image measures of Brownian motion. These estimates are established in Section 6.
 
The results of this paper enable one to to define the an anologue of L\'evy local time for algorithmically random Brownian motion. This line of thought will be discussed in a sequel to this paper.

It would also be interesting to look at  at the beautiful results of \L aba and Pramanik \cite{LaPr:1} within the context of this paper.
Some results  along these lines can be found in \cite{fo:7}.
\subsection*{Acknowledgements}
Our research has  been supported by the European Union grant agreement Marie Curie Actions - 
People International Research Staff Exchange Scheme  
PIRSES-GA-2011- 294962  in Computable Analysis (COMPUTAL).   The first author has also been supported by the National Research Foundation (NRF) of South Africa and the second author is  supported by the Vision Keepers' programme of the University of South Africa (UNISA).

The  leader of the COMPUTAL-project  is  Dieter Spreen to  whom this paper is dedicated. We are privileged  to have him as  our colleague and we are in the pleasant position to look  forward to further collaboration and a fruitful exchange of ideas with him and as a natural consequence, with colleagues involved in the project.

Many thanks are due to the members  of the Mathematics Department of the Corvinus University, Budapest,  for hosting the first author on many occasions and for  their constructive and critical  discussions on the results in this paper. 

The authors would  like to thank Mr and Mrs de Wet for generously making
available their farm Olmuja in KwaZulu Natal for accommodating research weeks which
 contributed greatly to the development of the results in this paper.

We appreciate the incisive comments from the referees of this paper. 

\section{Preliminaries from Brownian motion, Kolmogorov Complexity and Geometric Measure theory}

\subsection{Brownian motion and complex oscillations}

 Suppose we are given a probability space $\left(\Omega,P,\mathcal{A}\right)$, and a set $E$. A \emph{random element} of $E$, or \emph{random object} is a mapping $X$ from $\Omega$ into $E$.  The usual problem is to consider a subset $B$ of $E$ and to determine the probability of the subset of $\Omega$  consisting of the $
\omega$  such that $X(\omega)\in B$.  Of course, this only makes sense when $X^{-1}(B)\in \mathcal{A}$.  Then, instead of $X^{-1}(B)$, or $\left\{\omega : X(\omega)\in B\right\}$, we simply write $\left[X \in B\right]$, and we can speak of the event: ``$X$ belongs to $B$''.  In a similar vein, when $X \in B$ is defined by a predicate $Q\left(X\right)$ on $E$, we write $\left[Q\left(X\right)\right]$ for $\left[X \in B\right]$ and refer to it as the event defined by the predicate $Q$.  If this event holds with probability one, we say that ``$Q\left(X\right)$ holds almost surely'' and write $Q\left(X\right)$ a.s.  An important case is when $E$ is a topological space.  We say that a mapping $X$ from $\Omega$ to $E$ is a \emph{a random variable in E} if $X^{-1}(G)\in \mathcal{A}$ for all open subsets $G$ of $E$.  This has the implication that for each $B$ of the $\sigma$-algebra $\mathcal{B}$ generated by the open set in $E$ (the ``\emph{Borel algebra''} of $E$), the set $X^{-1}(B)$ belongs to $\mathcal{A}$. The mapping $B \mapsto P\left(X^{-1}\left(B\right)\right)$ defines a probability measure $\nu_X$ on $\left(E,\mathcal{B}\right)$, which is called the \emph{distribution} of $X$.  Two random variables on a probability space are said to be \emph{statistically equivalent} when their distributions are the same. From the point of view of integration theory, the measures $P$ and $\nu_X$ are related as follows:  Suppose $f$ is a real-valued function on $E$ such that $f(X)$ is integrable with respect to $P$.  Then

     $$\int_{\Omega}f\left(X\left(\omega\right)\right)dP(\omega)=\int_{E}f(x)d\nu_{X}(x).$$
This is known as the \emph{change of variable formula}.  If $X$ is a real-valued random variable on $(\Omega,P,\mathcal{A})$ such that its distribution $\nu_{X}$ is absolutely continuous with respect to Lebesgue measure $\lambda$ and if $f$ is the Radon-Nikodym derivative of $\nu_{X}$ with respect to $\lambda$, we say that $X$ has a density and call $f$ the \emph{density function} of $X$.  In this case, for a Borel set $A$ of real numbers, we have:

$$P\left(X\in A\right)=\nu_{X}(A)=\int_{A}f(t)dt.$$

A random variable $X$ with mean $m$ and non-zero variance $\sigma^2$ is {\it normal} if it has a density function of the form 
\[\frac{1}{\sqrt{2\pi}\;\sigma}\;e^{-(t-m)^2/2\sigma^2}.\]

A Brownian motion on the unit interval is a real-valued function $(w,t) \mapsto X_w(t)$ on $\Omega \times [0,1]$, where $\Omega$ is the underlying space of some probability space, such that $X_w(0)=0$ a.s.\ and for $t_1 < \ldots < t_n$ in the unit interval, the random variables $X_w(t_1),X_w(t_2)-X_w(t_1), \cdots,
X_w(t_n)-X_w(t_{n-1})$ are statistically independent and normally distributed with means all $0$ and variances $t_1,t_2-t_1,\cdots,t_n-t_{n-1}$, respectively.
This means that the probability of a finitary event of the form $\left[X(t_{j})\in A_{j} \;\mbox{for}\; 1 \leq j \leq n\right],$ where $0<t_{1}<\ldots < t_{n}\leq 1$, will be given by the integral

 $$\int_{A_{1}}\ldots\int_{A_{n}}\prod^{n}_{j=1}\frac{1}{\sqrt{2\pi\left(t_{j}-t_{j-1}\right)}}\mbox{exp}
     \left[\frac{-\left(y_{j}-y_{j-1}\right)^{2}}{2\left(t_{j}-t_{j-1}\right)}\right]dy_{n}\ldots dy_{1},$$
where $t_{0}=0$ and $A_{1},\ldots, A_{n}$ are Borel subsets of the reals.
It is a fundamental fact that any Brownian motion has a ``continuous version''(see, for example \cite{fre:1}). This means the following:
 Write $C[0,1]$ for the set of real-valued continuous functions on the unit interval and $\Sigma$ for the $\sigma$-algebra of Borel sets of $C[0,1]$ where the latter is topologised by the uniform norm topology $$\|x\|_\infty = \sup_{0\leq t\leq 1}|x(t)|.$$ There is a unique probability measure $W$ on $\Sigma$
such that for $0\leq t_1 < \ldots <t_n \leq 1 $ and for a Borel subset $B$ of $\mathbf{R}^n$, we have
\[ P(\{w \in \Omega:(X_w(t_1), \cdots, X_w(t_n)) \in B \}) = W(A) ,\]
where
\[A =\{x \in C[0,1]: (x(t_1), \cdots, x(t_n)) \in B\}.\]
The measure $W$ is known as the {\it Wiener measure}. We shall usually write $X(t)$ instead of $X_w(t)$. 

The set of finite words over the alphabet $\{0,1\}$ is denoted by $\{0,1\}^*$.
 If $a \in \{0,1\}^*$, we write $|a|$ for the length of $a$. If $\alpha=\alpha_0\alpha_1\ldots$ is an infinite word over the alphabet $\{0,1\}$, we write $\overline{\alpha}(n)$ for the word $\prod_{j<n}\alpha_j$. 
We use the usual recursion-theoretic terminology $\Sigma_r^0$ and $\Pi_r^0$ for the arithmetical subsets of $\mathbb{N}^k \times \{0,1\}^{\mathbb{N}\times l},\;k,l \in \mathbb{N} \). 
(See, for example, \cite{hin:1}). We denote the Cantor space $\{0,1\}^\mathbb{N}$ by $\mathcal{N}$. We write $\lambda$ for the Lebesgue probability measure on $\{0,1\}^{\mathbb{N}}$. For a binary word $s$ of length $n$, say, we write $[s]$ for the ``interval'' $\{\alpha \in \{0,1\}^\mathbb{N}: \overline{\alpha}(n) = s \}$. 
 A sequence $(a_n)$ of real numbers converges {\it effectively} to $0$ as $n \rightarrow \infty$ if for some total recursive $f:\mathbb{N} \rightarrow \mathbb{N}$, it is the case that $|a_n| \leq (m+1)^{-1}$ whenever $n \geq f(m)$. A real number $\beta$ is said to be recursive (or computable) if there is an algorithm which yields, for every natural number $n$, two integers $p,q$ such that $|\beta - \frac{p}{q}| < \frac{1}{2^n}$.

If $f,g$ are positive real-valued functions on $\{0,1\}^*$, we write $f< ^+g$ to signify that there is an absolute positive constant $D$ such that $f(x) <g(x) +D$ for all values of $x$. By a slight but useful abuse of notation we shall sometimes write $f(x) <^+g(x)$ instead of $f<^+g$. 

Let $(\phi_e: e \geq 0)$ be an effective enumeration of all the partial recursive functions from from $\{0,1\}^*$ to $\{0,1\}^*$. Let $U$ be the partial recursive function given by 
$$U(0^{e-1}1\sigma)=\phi_e(\sigma),\; e \geq 1,\; \sigma \in \{0,1\}^*.$$

In this case, we call $0^{e-1}1\sigma$ a program (or description) for $U$ of the number $s=\phi_e(\sigma)$.
For $s$ a finite binary string, the {\it plain} $U$-descriptive  complexity of $s$,  denoted by  $C_U(s)$, is the length of the shortest program for $U$ which will output $s$.

 In the sequel we write $C(s)$ in stead of $C_U(s)$. The values of $C_U$ are independent of the choice of the effective enumeration of the partial recursive functions $(\phi_e)$ up to the relation $<^+$.  (See \cite[pp 75-79]{nie:1} for a thorough discussion.)

A  universal prefix-free machine $V$ is just a universal Turing machine with domain a prefix-free set. This means that for no two distinct words in the domain of $V$ can the one be an initial segment of the other.  We call $V$ optimal if $C_V <^+ C_V'$ for every universal prefix-free universal Turing machine $V'$. The construction of an optimal universal prefix-free Turing machine can be found in \cite[p 84]{nie:1}, for example. Let $V$ be an optimal  universal prefix-free Turing machine from $\{0,1\}^*$ to $\{0,1\}^*$. For $s$ a finite binary string, the prefix-free Kolmogorov complexity $ K_V(s)$ of $s$ is the length of the shortest program for $V$ which outputs $s$.  In the sequel, we fix $V$ and denote the associated Kolmogorov complexity by $K(s)$. For a machine-independent modulo $<^+$  characterisation of $K$, see Theorem 2.2.19 in \cite[p 90]{nie:1}.

It is well-known that $C$ and $K$ are related as follows:
\begin{equation}
C(x) <^+ K(x) <^+ C(x) + 2 \log(C(x)) <^+ C(x) + 2 \log(|x|).
\label{eq:plainversusprefixfree}
\end{equation}
(See, for example, p 94 in \cite{nie:1}.)

Recall that an infinite binary string $\alpha$ is Kolmogorov-Chaitin complex if 

\begin{equation}
\exists_d\; \forall_n\;K(\overline{\alpha}(n)) \geq n-d.
\label{eq:Kolmogorov}
\end{equation}
In the sequel, we shall denote the set of Kolmogorov-Chaitin binary strings by $KC$ and refer to its elements as $KC$-strings. The set $KC$ is independent of our choice of the the optimal universal prefix-free machine $V$. (See, e.g., \cite{ch:2}, \cite{Do:01}, \cite{marlo:1}, \cite{lev:1} or \cite{nie:1} for more background.)


For $n \geq 1$, we write $C_n$ for the class of continuous functions on the unit interval that vanish at $0$ and are linear with slopes $\pm \sqrt{n}$ on the intervals $[(i-1)/n,i/n]\;,i=1,\ldots ,n$. With every $x \in C_n$, one can associate a binary string $a =a_1\cdots a_n$ by setting $a_i=1$ or $a_i =0$ according to whether $x$ increases or decreases on the interval $[(i-1)/n,i/n]$. We call the word $a$ the code of $x$ and denote it by $c(x)$. Conversely, any binary string $a =a_1\cdots a_n$ uniquely determines a function $x\in C_n$ and it will be denoted $x = c_*(a)$. 

The following notion was introduced by Asarin and Pokrovskii in \cite{ap:1}.
\begin{defi}

A sequence $(x_n)$ in $C[0,1]$ is \emph{complex}  if $x_n \in C_n$ for each $n$ and there is a constant $d > 0$ such that 
\begin{equation}
K(c(x_n)) \geq n-d
\label{eq:incompressibility}
\end{equation}
 for all $n$. A function $\phi  \in C[0,1]$ is a \emph{complex oscillation} if there is a complex sequence $(x_n)$ such that $\|\phi -x_n\|_\infty$ converges effectively to $0$ as $n \rightarrow \infty$.\
\label{definition:ap}
\end{defi}

The class of complex oscillations is denoted by $\mathcal{C}$. It is well-known that the class of complex oscillations has Wiener measure one. For another effective measure-theoretic characterisation of $\mathcal{C}$, the reader is referred to \cite{fo:1}.\\

REMARK. If $\varepsilon$ is a $KC$-string, then $(c_*(\bar{\varepsilon}(n)))$ is a complex sequence. One readily sees that $c_*(\bar{\varepsilon}(n))$ diverges in $C\left[0,1\right]$ \cite{fo:1}. Indeed, this result follows from the simple probabilistic observation to the effect that for $\alpha = \Pi\alpha_{j}$ in $\mathcal{N}$, the sequence $\left(\left(\alpha_{0}+\ldots+\alpha_{n}\right)/\sqrt{n+1}\right))$ diverges almost surely.  This follows, for example, from Khintchine's law of the iterated logarithm which, when applied to the probability space $(\mathcal{N},\mathcal{\lambda},\mathcal{B})$ (with $\mathcal{B}$ the Borel-algebra on $\mathcal\{0,1\}^\mathbb{N}$), states that
$$\limsup_{n\rightarrow \infty}\frac{\alpha_{0}+\ldots + \alpha_{n}}{\sqrt{2\left(n+1\right)\log \log\left(n+1\right)}}=1, $$
almost surely.  It then follows that all the $\Pi_{1}^{0}$ sets $B_{n}$ defined by
$$\alpha \in B_{n}\Longleftrightarrow  \forall_k \;\forall_\ell \left[k,\ell>n\Rightarrow||c_*(\bar{\alpha}\left(k\right))-c_*(\bar{\alpha}\left(\ell\right))||_\infty <1\right]$$
are of $\lambda$-measure 0 and, therefore, that no set $B_{n}$ can contain any $KC$-string. 

Here we have worked with a  Kurtz-test and invoked the results to be found on \cite[p 127] {nie:1} for example. What we required is the fact that no $\Pi_{1}^{0}$-set in $\mathcal{N}$ of $\lambda$-measure zero will contain any $KC$-string. Or equivalently, if a $\Pi_{1}^{0}$-event happens to contain a $KC$ element it is probabilistically significant to the extent that the event has non-zero probability.

In particular, for no $KC$-string $\epsilon$, will the initial segments $\overline{\epsilon}(n)$ define a complex oscillation via the operation $c_*$. Fortunately, there is a computational way of interpreting complex oscillations in terms of $KC$ (see Theorem \ref{theorem:isomorphism}  \cite{fo:2}), an observation which will play a central r\^ ole in this paper.\\  

In the sequel, if $\phi$ is a complex oscillation ($\phi \in \mathcal{C}$), and $d$ is a natural number such that for some complex sequence $x_n$  satisfying (\ref{eq:incompressibility}) converges effectively to $\phi$, we shall call $d$ an {\it incompressibility coefficient} of $\phi$.

  For recent refinements of the result of Asarin and Pokrovskii, the reader is referred to the work of  Kjos-Hanssen and Szabados \cite{KHSZ:01}.  They note  that Brownian motion and scaled,  interpolated simple random walks can be
jointly embedded in a probability space in such a way that almost surely, the $n$-step walk is, with respect to the uniform norm, 
within a distance $ O(n^{-\frac{1}{2} } \log n)$ of the Brownian path, for all but finitely many 
positive integers $ n$.  In the same paper, Kjos-Hanssen and Szabados show that, almost surely, their constructed  sequence $(x_n)$ of $n$-step walks is complex  in the sense of  Definition  \ref{definition:ap} and all Martin-L\" of random paths  (= complex oscillations)  have such
an incompressible close approximant. This strengthens a result of Asarin \cite{as:1}, who obtained
instead the bound $O(n^{-\frac{1}{6}} \log n)$. 

 The following result will play an important r\^ ole in this paper:

\begin{thm} \cite{fo:2}. There is a bijection $\Phi:KC \rightarrow\mathcal{C}$ and a  uniform algorithm that, relative to any $KC$-string $\alpha$, with input a dyadic rational number $t$ in the unit interval and a natural number $n$, will output the first $n$ bits of the the value of the complex oscillation $\Phi(\alpha)$ at  $t$. 
 \label{theorem:isomorphism}
\end{thm}

 Fouch\'e \cite{fo:4} proved that every complex oscillation $\phi$ has the following modulus of continuity:
  For any $C>1$, and for all sufficiently small values of $h$, 
    \begin{eqnarray} \label{aasdaadaga}
    \sup_{t\in [0, 1]} \vert \phi(t +h) - \phi(t)\vert \leq \sqrt{2 C \vert h \vert \log(1/\vert h \vert)}.
    \end{eqnarray}
As has been noted before, Kjos-Hanssen and Szabados \cite{KHSZ:01} showed that for any  complex oscillation $\phi$ there exists a complex sequence $(x_n) $ in $C[0, 1]$ and a constant $C$ such that, for $n$ sufficiently large,
\begin{eqnarray} \label{123dadgasa}
\Vert \phi - x_n\Vert _\infty\leq \frac{C \log n}{\sqrt{n}}. \end{eqnarray}

\subsection{Hausdorff and Fourier dimensions} \label{dadaer3tee}
If $(X,d)$ is a metric space, a {\it regular} Borel measure $\mu$ on $X$ is a Borel measure with the property that for every subset $A$ of $X$, there is a Borel set $B$ containing $A$ such that $\mu(B)=\mu^*(A)$. Here $\mu^*$ is the outer measure associated with $\mu$, in other words, writing $\mathcal{B}$ for the Borel algebra on $X$, 
$$\mu^*(A) = \inf \{\mu(B): A \subset B \in \mathcal{B}\}, \; A \subset X.$$ A measure on the Euclidean space $\mathbf{R}^d$ is a {\it Radon measure} if it is a  regular Borel measure and assumes finite values on compact subsets of $\mathbf{R}^d$, i.e., if it is locally finite. On a general complete separable metric space (a Polish space),  a measure is Radon iff it is locally finite and Borel regular.

If $X$ is a Polish space, and $\mu$ is a Radon measure on $X$, the support of $\mu$, denoted by $\mbox{supp}(\mu)$, is the complement of the set of points of $x \in X$ such that $\mu$ vanishes on some neighbourhood of $x$. Hence the support of $\mu$ is the smallest closed set $F$ such that $\mu(F^c)=0$, where $F^c$ denotes the complement of $F$ in $X$.

If $X$ and $Y$ are Polish spaces, $f$ is a Borel mapping from $X$ to $Y$, and $\mu$ is a Radon measure on $X$, the pushout measure $f_*\mu$ is given by 
$$f_*\mu(A) = \mu(f^{-1}(A)),$$ for Borel sets $A$ in $Y$. It is known that if $f$ is continuous and  the Radon measure $\mu$ has compact support then $f_*\mu$ too is a Radon measure. Moreover
$$\mbox{supp}f_*\mu = f(\mbox{supp}(\mu)).$$
For any Borel function $g:Y \rightarrow \mathbf{C}$, the change of variable formula is given by $$\int_Yg(y)df_*\mu(y)= \int_X g(f(x))d\mu(x).$$

We shall frequently apply this in the following way: Let $X:[0,1] \rightarrow \mathbf{R}$ be a continuous function and let $\mu$ be a non-zero Radon measure on $[0,1]$. If we write $\nu$ for the pushout of $\mu$ under $X$, then for all reals $\xi$  

\begin{equation}
 \int_\mathbf{R} e^{i s \xi} d\nu(s) = \int_0^1 e^{iX(t)\xi }d\mu(t).
\label{eq:fourier}
\end{equation}
This formula expresses the Fourier transform $\hat\nu(\xi)$ in terms of $X$ and $\mu$.

For a compact subset $E$ of Euclidean space $\mathbf{R}^d$ and real numbers $\alpha,\epsilon$ with $0 \leq \alpha < d$ and $\epsilon > 0$, consider all coverings of $A$ by balls $B_n$ of diameter $\leq \epsilon$ and the corresponding sums \[\sum_n|B_n|^\alpha,\]
where $|B|$ denotes the diameter of $B$. All the metric notions here are to be understood in terms of the standard $\ell^2$ norms on Euclidean space. The infimum of the sums over all coverings of $E$ by balls of diameter $\leq \epsilon$ is denoted by $H_\alpha^\epsilon(E)$. When $\epsilon$ decreases to $0$, the
corresponding $H_\alpha^\epsilon(E)$ increases to a limit (which may be infinite). The limit is denoted by $H_\alpha(E)$ and is called the Hausdorff measure of $E$ in dimension $\alpha$ or $\alpha$-Hausdorff measure of $E$.

If $0 < \alpha < \beta \leq d$, then, for any covering $(B_n)$ of $E$,
\[\sum_n|B_n|^\beta \leq \sup_n|B_n|^{\beta-\alpha}\sum_m|B_m|^\alpha,\]
from which it follows that
\[ H_\beta^\epsilon(E) \leq \epsilon^{\beta - \alpha}H_\alpha^\epsilon(E).\]
Hence if $H_\alpha(E) < \infty$, then $H_\beta(E)=0$. Equivalently,
\[ H_\beta(E) > 0 \Longrightarrow H_\alpha(E)=\infty.\]
Therefore,
\[ \sup\{\alpha:H_\alpha(E)=\infty\}=\inf\{\beta:H_\beta(E)=0\}.\]
This common value is called the Hausdorff dimension of $E$ and denoted by $\dim_h E$.

If $\alpha$ is such that $0 < H_\alpha(E) < \infty$, then $\alpha = \dim_h E$. However, $\alpha = \dim_h E$ does not necessarily imply $0 < H_\alpha(E) < \infty$. 



\begin{prop}
{\bf(Frostman's lemma)} \cite{frostman:1}. If $E$ is a compact subset of $\mathbf{R}^d$, then $H_\alpha(E) > 0$ if and only if  $E$ carries a nonzero Radon measure $\nu$ such that 
\[\nu(B) \leq C|B|^\alpha,\]
for all balls $B$ in $\mathbf{R}^d$ and some constant $C > 0$. 
\label{frostman}

\end{prop}

Let $E$ be a non-empty closed subset of the unit interval such that $\dim_h E \geq \alpha$. For $\beta < \alpha$, let $\nu$ be a non-zero Radon measure such that for some constant $C > 0$,
\begin{equation}
\nu(I) \leq  C|I|^\beta
\label{eq:massdistribution}
\end{equation}
 for all subintervals of the unit interval. It is the case that $\nu$ can be chosen such that 
$\nu(I) \leq |I|^\beta$ for all dyadic subintervals of the unit interval and $C = 3$ in (\ref{eq:massdistribution}). 

With the pair $(E ,\nu)$, we associate a directed tree $T=T(E)$ and a flow $f$ thereon as follows: The vertex set $V$ of $T$ consists of all dyadic intervals $I$ such that $E \cap I \neq
\emptyset$. Denote these intervals as $I_{j,n} = [(j-1)/N, j/N)$ for $1\leq j \leq N-1$ and $I_{j,n} = [(j-1)/N, j/N]$ for $j = N$ where $N = 2^n$. The edges are of the form $J_1 \rightarrow J_2$ where $J_1,J_2$ are in $V$ and $J_2$ is obtained upon dissecting $J_1$ (that is, if
$J_1 = I_{j,n}$, then $J_2 = I_{j,n+1}$ or $J_2 = I_{j+1, n+1}$). Write $\mathcal{E}$ for the set of edges of $T$. Define a function $f:\mathcal{E} \rightarrow \mathbf{R}^+$ by setting $f(J_1 \rightarrow J_2)=\nu(J_2)$. For $v \in V$, write $B_v$ for the set of edges emanating from $v$. Since $\nu$ is (finitely)
additive, it follows that $f$ is a {\it flow} on the tree. This means that for every $v \in V$ (such that $v$ is not the root), if we write $e_v$ for the edge ending in $v$, then, for all vertices $v$
\begin{equation}
 f(e_v)=\sum_{e \in B_v} f(e).
\label{eq:flow}
\end{equation}
Moreover, if $v$ is the dyadic interval of length $\frac{1}{N}$, then
\begin{equation}
 f(e_v) \leq \frac{1}{N^{\beta}},
\label{eq:flowbound}
\end{equation}

Conversely for a given compact set $E$, suppose that there is a flow $f$ on $T(E)$ satisfying (\ref{eq:flow}) and (\ref{eq:flowbound}). Then there is a non-zero Radon measure $\nu$ supported by $E$ such that (\ref{eq:massdistribution})
holds. To see this, for a dyadic interval $I$, set $\nu(I)=f(e)$, when $I$ is a vertex of the tree and $e$ is the edge ending in $I$, and set $\nu(I)=0$ if $I$ is not an edge, that is, if $I$ is disjoint from $E$. The flow-condition ensures that $\nu$ is countably additive on the semi-algebra of dyadic intervals. It
can therefore be extended to the Borel $\sigma$-algebra on the unit interval. It is quite readily seen  that for an arbitrary closed interval $I$ of the unit interval,
\[\nu(I) \leq 3|I|^\beta.\]
(This follows the fact for any interval $I \subset [0, 1]$, there exist 3 dyadic intervals $K_1, K_2, K_3$ of $[0, 1]$ of equal length $\leq |I|$ such that $I \subset K_1 \cup K_2 \cup K_3$.)

The idea of presenting Radon measures on the reals as flows on directed trees belongs to the folklore of fractal geometry. The authors learnt these ideas from the book by  M\"orters and Peres \cite{MoPe:01}. 
 
Denote by $\mathbb{D}$ the set of pairs $(j, n), \, 1 \leq j \leq N = 2^n$ ($j, n$ integers). 
To say that the function $(j, n) \mapsto \nu(I_{j,n})$ (defined on $\mathbb{D}$) is computable is equivalent to requiring that the flow $v \mapsto f(e_v)$ on the vertex set of the tree $T(E)$ associated with $(E,\nu)$ is computable. It is an interesting problem to identify conditions on a set $E$ of a given Hausdorff dimension to ensure the existence of a computable flow on the tree associated with $E$ that will provide computable measures $\nu$ witnessing the validity of Frostman's lemma. As shown in \cite{MoPe:01}, this problem is related to finding computable versions of the mincut-maxflow theorem on countable trees. We shall discuss a large class of perfect sets (Cantor sets) for which such constructive flows can be found.

For a compact set $E$ and a measure $\nu$ satisfying the conclusion of Proposition \ref{frostman}, let 
$$\nu_n = \sum_{j=1}^{N} \nu(I_{j,n})\delta_{j/N},\,\,\,(N = 2^n)$$
where $\delta_{j/N}$ is the Dirac measure concentrated at $j/N$. (We also frequently invoke Riesz's representation theorem and thus also think of  of Radon measures on a  compact Hausdorff space as positive linear functionals on the space $X$.)
The measures 
$\nu_n$ converge weakly to $\nu$ in the sense that for all $f \in C(E)$, the Banach space of continuous functions on $E$ with the uniform norm, it is the case that 
$$ \lim_{n \rightarrow \infty}\nu_n(f) =\nu(f),$$
for all $f \in C(E)$.

 The discrete measure $\nu_n$ will be called the $n$-approximation of $\nu$. It is clear that for any subinterval $I$ of $[0, 1]$, it is the case that
  \begin{eqnarray} \label{eadagasdw}
\nu_n(I) \leq 1/N^\alpha \mbox{ for } |I| < 1/N \mbox{ and } \nu_n(I) \leq 3 |I|^{\alpha} \mbox{ for } |I| \geq 1/N.
\end{eqnarray}
Indeed if $|I| < 1/N$, then $I \subset I_{j,n}$ for some $j$ and hence $$\nu_n(I) \leq \nu_n(I_{j,n}) = \nu(I_{j,n}) \leq |I_{j,n}|^\alpha = 1/N^\alpha.$$ For $|I|\geq 1/N$, consider 3 dyadic intervals $K_1$, $K_2$, $K_3$ of equal length $ r \leq |I|$ and $r \geq 1/N$ such that $I \subset K_1 \cup K_2 \cup K_3$. Then 
\begin{eqnarray*}
\nu_n(I)\leq \nu_n(K_1 \cup K_2 \cup K_3) & = &\sum_{j=1}^{N} \nu(I_{j,n})\delta_{j/N}(K_1 \cup K_2 \cup K_3)\\
&\leq & \nu_n(K_1) + \nu_n(K_2) + \nu_n(K_3)\\
&= & \nu(K_1) + \nu(K_2) + \nu(K_3)\\
&  \leq & |K_1|^{\alpha} + |K_2|^{\alpha}  +|K_3|^{\alpha} \\
& \leq&  3 |I|^\alpha.
\end{eqnarray*}

\begin{defi} Let $0 < \alpha < 1$. We call a Radon (probability)
 measure on the unit interval an $\alpha$-\emph{Frostman measure} if 
for some constant $C >0$, we have $\nu(I) \leq C|I|^\alpha$ for all dyadic intervals $I$ contained in the unit interval.
\end{defi}
This has the implication that $\nu(I) \leq  3C |I|^\alpha$ for any general interval $I$. 
 Recall that  $\mathbb{D}$ the set of pairs $(j, n), \, 1 \leq j \leq 2^n$ . For any $d = (j, n) \in \mathbb{D}$, denote $I_d = [(j-1)/2^n, j/2^n)$. 

\begin{defi}
An $\alpha$-Frostman measure $\nu$ is called an \emph{effective} $\alpha$-Frostman measure if the function $$\mathbb{D} \rightarrow \mathbf{R}, \,\,d \mapsto \nu(I_d) $$ is computable. 
\end{defi}
This means that one can effectively and uniformly find, from $d \in \mathbb{D}, m\in \mathbb{N}$, some rational $r$ such that $\vert \nu(I_d) - r \vert \leq 1/2^m.$

\begin{defi}
Let $E$ be a compact subset of Hausdorff dimension $\beta > 0$. We say that the Hausdorff dimension of $E$ is \emph{ effectively witnessed}, if for each rational $0 < \alpha < \beta$, there is an effective $\alpha$-Frostman measure which supports $E$.
\end{defi}

 Cantor ternary sets $C_\xi$ with a computable ratio $0< \xi <1/2$ give examples of  compact subsets of the unit interval which effectively witness the Hausdorff dimensions $$\beta = \log 2/\log(1/\xi)$$ of these sets. We recall that the sets are  constructed in the following way. Start from the interval $[0, 1]$, remove an open interval of length $1- 2\xi$ in the middle of
the original interval $[0, 1]$, then from each of the two remaining intervals, remove the interval of length $\xi(1- 2\xi)$ in the middle. At the $n$th step, there are $2^n$ closed intervals of common length $\xi^n$ and each of these generates two subintervals of length $\xi^{n+1}$ by removing an open interval of
length $\xi^n(1-2 \xi)$ in the middle. Denote by  $\mathcal{A}_n$ the set of the $2^n$ intervals that have survived up to  stage $n$. Each of these intervals has length $\xi^n$. By definition,  
\begin{equation}
C_\xi = \bigcap_{n=1}^\infty \bigcup_{J \in \mathcal{A}_n} J.
\label{equation:cantor}
\end{equation}
 It is well-known that the $\beta$-Hausdorff measure of $C_\xi$ is $1$. This has the implication that $$\dim_hC_\xi =\beta.$$
One can construct a non-zero Radon measure
$\nu$ supported by $C_\xi$ such that $$\nu(J) = 1/2^n=|J|^\beta,$$ for each $J \in \mathcal{A}_n$. 

It is well-known that
$$\nu(I) \leq C |I|^\beta,$$
for all intervals $I$ and a constant $C$. (See, for example \cite[p 62]{Ma:01}.)

Given a dyadic interval $I$ and a natural number $n \geq 1$, since $\xi$ is computable, we can effectively compute the number $L$ of intervals $J \in \mathcal{A}_n$ which intersect $I$. Note that
$$\left|\nu(I) - \frac{L}{2^n}\right| \leq \frac{2}{2^n}.$$ We conclude that $\nu$ is an effective $\beta$-Frostman measure and that the Hausdorff dimension $\beta$ of $C_\xi$ is effectively witnessed.

For $0 < \xi <\frac{1}{4}$ set $\gamma =2\beta$ where $\beta$ is the Hausdorff dimension of $C_\xi$. Note that $\gamma$ is computable iff $\xi$ is and will range over all the computable reals in the (open) unit interval as $\xi$ ranges over the computable reals in $(0,\frac{1}{4})$. For $\omega \in KC$ and a computable $\xi$ in the unit interval set
\begin{equation}
S_\gamma(\omega)= \Phi(\omega)(C_\xi),\; \gamma =2 \frac{\log2}{\log\frac{1}{\xi}},
\label{equation:theset}
\end{equation}
 where $\Phi(\omega)$ is the complex oscillation asssociated with $\omega$. (See Theorem \ref{theorem:isomorphism}.)

Note that unfolding (\ref{equation:cantor}) yields
\begin{equation}
x \in S_\gamma(\omega) \Leftrightarrow  \forall_n\exists_{y \in \mathbb{D}_0}\exists_{J \in \mathcal{A}_n} \;(y \in J ) \;\wedge |\Phi(\omega)(y) - x| < \frac{1}{2^n},
\label{equation:salemfromkc}
\end{equation}
where $\mathbb{D}_0$ denotes the set of dyadic numbers in the unit interval. 

 It follows from Theorem \ref{theorem:isomorphism} and the fact that $\gamma$ hence $\xi$ is computable that the sets $S_\gamma(\omega)$ are $\Pi_2^0(\gamma, \omega)$ definable over $(0,	1)_r \times KC$, where $(0,1)_r$ denotes the set of computable real numbers 

For a Radon measure $\mu$ on $\mathbf{R}^d$  with compact support set
\[I_\alpha(\mu) =\int_{\mathbf{R}^d}\int_{\mathbf{R} ^d}\frac{d\mu(x)d\mu(y)}{|x-y|^\alpha}.\]
We say that $\mu$ has finite energy with respect to $|x|^{-\alpha}$ when $I_\alpha(\mu) < \infty$. If $E$ carries positive measures of finite energy with respect to $|x|^{-\alpha}$ we say that $E$ has positive capacity with respect to $|x|^{-\alpha}$ and we write 
\[\mbox{Cap}_\alpha(E) > 0.\]
If $E$ carries no positive measure of finite energy with respect to $|x|^{-\alpha}$, we say that $E$ has capacity zero with respect to this kernel and we write $\mbox{Cap}_\alpha(E)=0$.
(See, for example Chapter 10 in \cite[pp 132-134]{kah:1} for more details on capacities.)

It follows from the Fourier analysis of Schwartz distributions that
\begin{equation}
I_\alpha(\mu)=C(\alpha,d)\int_{\mathbf{R}^d}|\hat\mu(u))|^2|u|^\alpha\frac{du}{|u|^d},
\label{equation:energy}
\end{equation}
when $0 < \alpha < d$, where $C(\alpha,d)$ is a positive constant and where,  moreover,
\[\hat\mu(u)=\int_{\mathbf{R}^d}e^{isu}d\mu(s),\]
the Fourier transform of the measure $\mu$.
For more on the Fourier analysis of Radon measures, see, for example, \cite[162-163]{Ma:01}. 

The following result is well-known. (See, for example \cite[p 133]{kah:1}.)
\begin{prop}
For a compact subset $E$ of $\mathbf{R}^d$ and $0 < \alpha < \beta < d$,
\[H_\beta(E) > 0 \Rightarrow \mbox{\emph {Cap}}_\alpha(E)>0 \Rightarrow H_\alpha(E) > 0, \]
and \[\sup\{\alpha:I_\alpha(\mu) < \infty\} = \sup\{\alpha:\mbox{\emph {Cap}}_\alpha(E) > 0\} = \mbox{\emph{dim}}_h E.\]
\label{proposition:hausdorfcapacity}
\end{prop}

\begin{defi}
Let $E$ be a compact subset of $\mathbb{R}^d$. The \emph{Fourier dimension} of $E$ (denoted $\dim_f E$) is the supremum of the numbers $0\leq \alpha \leq d$ such that $E$ carries a non-zero Radon measure $\mu$ such that the function $\mathbb{R}^d \rightarrow  \mathbb{R}^+: u \mapsto |u | ^{\alpha/2} | \hat \mu (u) |$,
$(u \in \mathbb{R})$ is bounded.
\end{defi}
It is well known that for any compact $E$ of $\mathbb{R}^d$, $\dim_f E \leq \dim_h E$. If $0< \dim_f E = \dim_h E$, then $E$ is called a Salem set. Every Salem set (or in general every compact set with non-zero Fourier dimension) is a set of multiplicity. Proofs of these results can be found in \cite{Ma:01}, Chapter 12. See also \cite{Muk:01} for a fully self-contained exposition of these results.

We also require the following result, a proof of which can be found on p 139 of \cite{kah:1}.
\begin{thm}
Suppose we are given a compact set $E$ in $\mathbb{R}^d$, and a mapping $f$ from $E$ to $\mathbb{R}^n$ such that for some $0 < \beta < 1$  and $C>0$,  $$|f(x) -f(y)| \leq C |x -y|^\beta,$$  for $x,y \in E$. Then $$\dim_h(f(E)) \leq \min(\frac{1}{\beta} \dim_h(E),n).$$
\end{thm}

By now taking (\ref{aasdaadaga}) into account, we can infer the following:
\begin{thm}
If $\phi$ is a complex oscillation and $E$ is a compact subset of the unit interval, then
$$\dim_h\phi(E) \leq \min(2\dim_h(E),1).$$
\label{theorem:imagedimension}
\end{thm}

{\section{Main Results}

We shall prove the following:

\begin{thm} \label{asdadgha1}
Let $0 < \alpha \leq 1$. Suppose $\phi$ is a complex oscillation and $\theta$ is an effective $\alpha$-Frostman measure on $[0, 1]$ and $\epsilon >0$. Then  for all reals $u$ such that $|u|$ is sufficiently large (depending on $\epsilon$), 
\begin{eqnarray} \label{eqadaqf087}
 \left|\int_0^1 e^{iu\phi(t)} d\theta(t)\right| \leq \frac{2}{|u|^{\alpha - \epsilon}}.
\end{eqnarray}
\end{thm}
 In other words, if  $\nu$ is the pushout measure of $\theta$ under the complex oscillation $\phi$, its Fourier transform satisfies
$$|\hat{\nu}(u)| \leq C \frac{1}{ |u|^{\alpha - \epsilon}}$$ for a constant C (depending only on $\epsilon$) and $|u|\geq 1$.
Now write $E$ for the support of $\theta$ and $F$ for the support of $\nu$. Then $F = \phi(E)$. It follows that $$\dim_f F \geq \min(1, 2\alpha).$$
This theorem has the following consequence:
\begin{thm}
Suppose $E$ is a compact subset of the unit interval whose Hausdorff dimension $\beta$ is effectively witnessed 
and such that $0 < \beta \leq 1$. If $\phi$ is a complex oscillation, then the image, $\phi(E)$, is a Salem set of dimension $\min\{1, 2 \beta\}$.
\end{thm}

\begin{proof}
By Theorem \ref{asdadgha1}, the Fourier dimension of $\phi(E)$ is $\geq  \min\{1, 2\alpha\}$ for any $\alpha$ such that $0\leq \alpha < \beta$. Hence $\dim_f \phi(E) \geq \min\{1, 2\beta\}$. 

Since $\dim_h \phi(E) \leq \min\{1, 2 \beta\}$ (Theorem \ref{theorem:imagedimension}), it follows that $\dim_f \phi(E) = \dim_h \phi(E) = \min\{1,  2\beta\}$ (using the fact that $\mbox{dim}_f \phi(E) \leq \mbox{dim}_h \phi(E)$).
\end{proof}

It follows that all the sets $S_\gamma(\omega)$  (see (\ref{equation:cantor}), (\ref{equation:theset}) and (\ref{equation:salemfromkc})) with $\gamma \in (0,1)_r$ and $\omega \in KC$, are Salem sets. Moreover
$$\dim_h S_\gamma(\omega) =\dim_f S_\gamma(\omega) =\gamma.$$ 
We can therefore infer:
\begin{thm}
Write $(0,1)_r$ for the set of real computable numbers in the unit interval and $KC$ for the set of infinite binary numbers which are random in the sense of Kolmogorov-Chaitin-Levin-Martin-L\" {o}f.
There is a $\Pi_2^0$ predicate over $\mathbf{R} \times (0,1)_r\times KC$ which for each $(\gamma,\omega) \in (0,1)_r\times KC$, defines a Salem set $S_\gamma(\omega)$ of Fourier dimension $\gamma$.
\label{theorem:effectivebeurling}
\end{thm}
\section{Fourier analytical properties of finite sequences of high Kolmogorov complexity}

In this section, we fix an integer $n$ and set $N = 2^n$. We assume that the set $\{0, 1\}^N$ is endowed with the canonical probability measure  (the product $\mu_1 * \mu_2 * \ldots *\mu_N$ where $\mu_k$ is the  Bernoulli  fair-coin toss measure on $\{0, 1\})$.

For any finite binary string $\omega \in \{0, 1\}^N$, denote by $S_n(\omega)$ the continuous function  on the unit interval that vanishes at $0$ and is linear on each interval $I_{j,n} = [(j-1)/N,j/N)$, $j=1,\ldots ,N$ with slope  $ \pm \sqrt{N}$ according to whether $\omega_j = 1$ or $\omega_j = 0$. Here $\omega  =\omega_0\omega_1\ldots$.
For any $t\in [0, 1]$, we denote by $S_n(t)$, the random variable $\omega \mapsto S_n(\omega, t) = S_n(\omega)(t)$ on $\{0, 1\}^N$. Recall that $\delta_{a}$ is the Dirac measure concentrated at $a$. 
We have the following result:
\begin{thm} \label{adjakfradad1234}
Let $\theta$ be an effective $\alpha$-Frostman measure on $[0, 1]$ and $\theta_n$ its $n$-step approximation, that is,
\begin{equation}
\theta_n = \sum_{j = 1}^N c(j) \delta_{j/N},
\label{eq:dirac}
\end{equation}
 where $c(j) = \theta(I_{j,n})$, \, $N=2^n$ and  $ I_{j,n} = [(j-1)/N, j/N))$ , for $j = 1, 2, \ldots, N$.
Then for any positive integer $q \leq \log n$, and any real number
$u$ such that $n \leq u \leq n+1$,
\begin{eqnarray}
\mathbb{E}\left(\left\vert \int_0^1 e^{i u S_n(\omega, t)} d \theta_n(t)\right\vert ^{2q}\right) \leq (22\,q\,u^{-2\alpha})^q ,
\end{eqnarray}
for all large values of $n$. 
\end{thm}
The proof of Theorem \ref{adjakfradad1234} is given in Section \ref{adafafki8}. 
We now deduce from Theorem \ref{adjakfradad1234} the following Theorem which will be used in the following section to prove Theorem \ref{asdadgha1}.

\begin{thm} \label{saf1}
Let $\theta$ be an effective $\alpha$-Frostman measure on $[0, 1]$ and $\theta_n$ its $n$-step approximation (\ref{eq:dirac}). 
Let $\epsilon >0$ be any fixed rational number.  For any natural number $d$, there is a number $L_d$ such that for $n \geq L_d$ we have for all $\omega \in \{0,1\}^N,\; N=2^n$, whose Kolmogorov complexity satisfies $$K(\omega) > N-d,$$ that
\begin{eqnarray} \label{dadjfad23}
\left \vert \int_0^1 e^{i u S_n(\omega, t)} d\theta_n(t)\right \vert^2 \leq  \frac{1}{u ^{2\alpha - \epsilon }},
\end{eqnarray}
  for all rational numbers $u = n,\, n+\frac{1}{n},\, n+\frac{2}{n},\, \ldots, n+1.$
\end{thm}


\begin{proof}
For a given $n$ and  $u \in\{ n,\, n+\frac{1}{n},\, n+\frac{2}{n},\, \ldots, n+1\}$, we define
$$F(\omega,u) = \left\vert \int_0^1 e^{i u S_n(\omega, t)} d \theta_n(t)\right\vert ^{2}, $$  when $\omega \in \{0,1\}^N$.

Then for any $\epsilon >0$, and sufficiently large integers $n$, we have that 
 \begin{eqnarray} \label{hasdqsADA}
 P \left\{F(\omega, u) >  u^{-2\alpha + \epsilon}\right \} \leq \frac{1}{u^5}.
 \end{eqnarray}
Indeed, by the Chebyshev's inequality and Theorem \ref{adjakfradad1234}, we obtain,  for $q \leq \log n$, that
\begin{eqnarray*} \label{hasdqsADasa}
 P \left\{F(\omega, u) >  u^{-2\alpha + \epsilon} \right \} & \leq &  \frac{\mathbb{E}\left[(F(\omega, u))^q\right]}{(u^{-2\alpha + \epsilon})^q}\\
 &\leq & \frac{(22\,q\,u^{-2\alpha})^q}{(u^{-2\alpha + \epsilon})^q}.
 \end{eqnarray*}
 Taking $q = \lceil 6/\epsilon \rceil$ (the smallest integer $\geq 6/\epsilon$) and $u$ such that $u \geq (22q)^{q}$ yields relation (\ref{hasdqsADA}).
 
Let us denote by $A_u$ the event  $F(\omega, u) > u^{-2\alpha + \epsilon}$ in $\{0, 1\}^N$. 
By relation (\ref{hasdqsADA}), for large $u$, 
 \begin{eqnarray} \label{eawddxgacfs}
 P \{A_u \} \leq \frac{1}{u^{5}}
\end{eqnarray}

We want to show that a finite binary $\omega \in A_{u}$ 
has low Kolmogorov-Chaitin complexity, since the cardinality of $
A_{u}$ is ``small". Consider the following algorithm $\phi _{1}$ which on
self-delimiting inputs for the integers $n$ and $r$ with $0\leq r\leq n$
enumerates the elements of $A_{u}.$ Firstly $\phi _{1}$ computes $N=2^{n}$
and $u=n+\frac{r}{n}.$ It then computes from $N$, increasingly accurate
approximations to the coefficients $c(1),c(2),\ldots ,c(N)$.  
Since $F(\omega, u)$ is computable in $c(1),c(2),\ldots ,c(N)$  we get that every string in $A_{u}$
will eventually be enumerated. 

We can thus specify any string in $A_{u}$ by giving its index $j$ in this
enumeration. For this we need only self-delimiting codes for $n,r$ and a
code for $j$. Since we know that $\left\vert A_{u}\right\vert \leq
2^{N-5\log n}$ we know that $\left\vert j\right\vert \leq N-5\log n.$ By
padding $j$ with initial zeroes if necessary, we can assume that 
\begin{equation*}
\left\vert j\right\vert =\left\lfloor N-5\log n\right\rfloor 
\end{equation*}%
making $j$'s length computable from input $n.$ Using the standard upper
bound of $2\log k$ for a self-delimiting code for $k,$ a program to generate 
$\omega $ will thus have length bounded by 
\begin{equation*}
2\log n+2\log r+(N-5\log n)+C,
\end{equation*}%
where $C$ covers our programming overheads. Clearly, for large enough $n,$
this drops below $N-d$ for any pre-given $d\in \mathbb{N}$. Hence, if $%
K(\omega )\geq N-d,$ then for large enough $n,$ $\omega \notin A_{u},$ or
equivalently, relation (\ref{dadjfad23})  holds.

 \end{proof}
  
  \section{Proof of Theorem \ref{asdadgha1}}
We are now ready to prove Theorem \ref{asdadgha1}. Assume that $(x_n)$ is a complex sequence which converges to $\phi$ as in (\ref{123dadgasa}). For $n \geq 1$ set $\phi_n=x_N$ where $N=2^n$. We also assume that  $K(\omega_n) \geq N -d $ for all $n$ and for a fixed constant $d$ where $\omega_n$ is the code of $\phi_n$. Note
that $\Vert \phi_n - \phi \Vert \leq \frac{C\, n}{\sqrt{N}}$ for a constant $C$ and all large $n$.  Then $S_n(\omega_n, t) = \phi_n(t)$ and Theorem \ref{saf1} implies that
  \begin{eqnarray} \label{dadj12fad23}
\left \vert \int_0^1 e^{i u \phi_n(t)} d\theta_n(t)\right \vert^2 \leq  \frac{1}{u ^{2\alpha - \epsilon }}
\end{eqnarray}
for large $n$ and all $u = n,\, n+\frac{1}{n},\, n+\frac{2}{n},\, \ldots, n+1.$

 Note that from relation (\ref{dadj12fad23}), one can deduce that, there exists a constant $C$, such that, for any real number $\xi$ with $ n\leq  \xi \leq n+1$, it is the case that, 
\begin{eqnarray} \label{eq:sfsge}
\left \vert \int_0^1 e^{i \xi \phi_n(t)} d\theta_n(t)\right \vert \leq  \frac{1}{\xi ^{\alpha - \epsilon }} +\frac{C}{n}. 
\end{eqnarray}
Indeed, for $n$ sufficiently large
\begin{eqnarray*}
\left \vert  \int_0^1 (e^{i u \phi_n(t)}  - e^{i \xi \phi_n(t)}) d\theta_n(t) \right\vert  & \leq & \vert  u - \xi\vert \int_0^1 \vert \phi_n(t) \vert d\theta_n(t)\\
 &\leq & C \vert  u - \xi\vert\
\end{eqnarray*}
where $C = \sup_{j\geq 1} \Vert \phi _j\Vert_\infty +1$. (Here we have used the  inequality $|e^{i x} - e^{i y}| \leq |x-y|$.)

Taking  $u$ such that $\vert u - \xi\vert <1/n$ together with (\ref{dadj12fad23}) yields (\ref{eq:sfsge}). 

Now for any real number $\xi$ such that $n \leq \xi \leq n+1$, we have that 
  \begin{eqnarray*}
   \left \vert  \int_0^1 (e^{i \xi \phi_n(t)}  - e^{i \xi \phi(t)}) d\theta_n(t) \right\vert  & \leq & \int_0^1 \xi \vert \phi_n(t) - \phi(t)\vert d\theta_n(t) \\
 & \leq & \xi\Vert \phi_n - \phi \Vert \mbox{ (since } \theta([0, 1]) \leq 1)\\
 &\leq & \frac{C_1\, n(n+1)}{\sqrt{N}}
   \end{eqnarray*}
(by relation (\ref{123dadgasa})) and the choice of $\xi$.) 
In addition,
\begin{eqnarray*}
\left\vert \int_0^1 e^{i \xi \phi} d\theta_n - \int_0^1 e^{i \xi \phi} d\theta \right\vert & = & \left\vert \sum_{j = 1}^N \int_{\frac{j-1}{N}} ^ {\frac{j}{N}} e^{i \xi \phi(j/N)} d\theta(t)  -   \int_{\frac{j-1}{N}} ^ {\frac{j}{N}} e^{i \xi \phi(t)}  d\theta(t)\right\vert\\
& \leq & \sum_{j = 1}^N  \int_{\frac{j-1}{N}} ^ {\frac{j}{N}} \left\vert e^{i \xi \phi(t)} - e^{i \xi \phi(t)} \right\vert d\theta(t) \\
& \leq & \xi (C_2 (1/N) \log N)^{1/2} \sum_{j=1}^N \theta([(j-1)/N, j/N]) \mbox{ (by } (\ref{aasdaadaga})) \\
& \leq & \frac{(n+1) \, \sqrt{C_2 \,n}}{\sqrt{N}}.
\end{eqnarray*}

To summarise, for $n$ sufficiently large and all real $\xi$ in the interval $(n,n+1)$

   \begin{eqnarray*}
 \left\vert  \int_0^1 e^{i \xi \phi(t)} d\theta(t) \right\vert  \leq \left\vert  \int_0^1 e^{i \xi \phi_n(t)} d\theta_n(t)  \right\vert +  \frac{C_1\, n(n+1)}{\sqrt{N}} +  \frac{(n+1) \, \sqrt{C_2 \,n}}{\sqrt{N}}. 
   \end{eqnarray*}
 It follows that for all $n$  large and all $\xi \in (n,n+1)$
  \begin{eqnarray*}
 \left \vert  \int_0^1 e^{i \xi \phi(t)} d\theta(t)\right \vert & = &  \frac{1}{\xi^{\alpha - \epsilon}} + O(\frac{1}{n}) \\
  &= &  \frac{1}{\xi^{\alpha - \epsilon}} + O(\frac{1}{\xi}).
\end{eqnarray*}
In conclusion, for sufficiently large  real number $\xi$, 
  \begin{eqnarray*}
\left  \vert  \int_0^1 e^{i \xi \phi(t)} d\theta(t) \right\vert   &\leq &  \frac{2}{|\xi|^{\alpha - \epsilon}}.
 \end{eqnarray*}
This extends obviously to negative $\xi$  with $\vert \xi \vert $ large.
   \qed

\section{Proof of Theorem \ref{adjakfradad1234}} \label{adafafki8}

We will need the following lemma, which is like an integration by parts for  singular measures. 
\begin{lem} \label{dakjasfad}
Let $\mu$ be the measure $\sum_{j=1}^n c_j \delta_{t_j}$ where $c_j$ are positive constants and  $0\leq t_1 < t_2 \ldots  t_n \leq 1$. 
Then for any differentiable function $f$ defined on $[0, 1]$,
\begin{eqnarray} \label{eq:sm2}
\int_0^1 f(t) d\mu(t) = \mu[0, 1] f(1) - \int_0^1 f\, {'}(t) \mu[0, t] dt.
\end{eqnarray}
\end{lem}

\begin{proof}
This can be proven by a direct calculation. 
Clearly,  $$\int_0^1 f(t) d\mu(t)  =  \sum_{j=1}^n c_j f(t_j)\mbox{ and  }
\mu[0, 1]f(1) = \sum_{j=1}^n c_j f(1).$$
Also
\begin{eqnarray*}
\int_0^1 f\, {'}(t) \mu[0, t] dt &  = &  c_1 \int_{t_1}^{t_2}  f\, {'}(t) dt + (c_1 + c_2) \int_{t_2}^{t_3}  f\, {'}(t) dt + \ldots  \\
&& +(c_1 + c_2 + \ldots + c_n) \int_{t_{n}}^{1}  f\, {'}(t) dt\\
& = & \sum_{j=1}^n c_j (f(1) - f(t_j)).
\end{eqnarray*}
The lemma follows immediately. 
\end{proof}

We are now ready to prove Theorem \ref{adjakfradad1234}. 

\proof
Since  
       $\theta_n = \sum_{j = 1}^N c(j) \delta_{j/N}$, one has that      
 $$\int_0^1 e^{i u S_n(t)} d \theta_n(t) = \sum_{j=1}^{N} c(j) e^{i n S_n(j/N)}.$$
Therefore
\begin{eqnarray*}
 \left \vert \int_0^1 e^{i u S_n(t)} d \theta_n(t)\right\vert ^{2q}  & = & \sum_{1\leq j_1, j_2, \ldots, j_{2q} \leq N} c(j_1) c(j_2) \ldots c(j_{2q}) \\ &&\times e^{i u (S_n(j_1/N) + \ldots + S_n(j_q/N))}
 e^{-iu((S_n(j_{q+1}/N) + \ldots + X_n(j_{2q}/N))}.
 \end{eqnarray*}
 By symmetry, this implies that 
 \begin{eqnarray*}
 \left \vert \int_0^1 e^{i u S_n(t)} d \theta_n(t)\right\vert ^{2q} & = & (q!)^2 \sum_{\epsilon \in T} \,\,\sum_{1\leq j_1\leq  \ldots\leq j_{2q} \leq 2^n} c(j_1) c(j_2) \ldots c(j_{2q}) \\ 
&&\times 
 e^{i u (\epsilon_1 S_n(j_1/N) + \ldots + \epsilon_{2q} S_n(j_{2q}/N))}
 \end{eqnarray*}
 where 
  $$T = \{\epsilon = (\epsilon_1, \ldots, \epsilon_{2q}): \epsilon_j = \pm 1 \mbox{ and } \epsilon_1 + \ldots+ \epsilon_{2q} = 0\}.$$
 We rewrite the sum $$\sum_{k=1}^{2q} \epsilon_k S_n(j_k/N)$$
as
 $$ \sum_{k=1}^{2q} \alpha_k [S_n(j_k/N) - S_n(j_{k-1}/N)]$$ where $\alpha_k = \epsilon_k + \ldots+ \epsilon_{2q} \mbox{ and } j_0 = 0. $
  
 Now note that, by definition, the random variables $$B_k = S_n(j_k/N) - S_n(j_{k-1}/N),\;0 \leq k \leq 2q,$$  are independent and each $B_k$ has the same  distribution as $S_n ((j_k - j_{k-1})/N)$. 
Since $S_n(h/N)$ is the sum of $h$ independent and identically distributed variables $V_1, V_2, \ldots, V_h$ with $P\{V_1 = 1/\sqrt{N}\} = P\{V_1 = -1/\sqrt{N}\} = 1/2$, we have that 
\begin{eqnarray*}
\mathbb{E}[e^{i u S_n(h/N)}] &= &\mathbb{E}[e^{i u (V_1 + V_2 + \ldots + V_h)}] \\
 &=& \prod_{j=1}^h \mathbb{E}[e^{i u V_j}] \mbox{ (by independence of the } V_j)\\
 &=& \left((1/2) e^{ i u/\sqrt{N}} + (1/2) e^{-iu/\sqrt{N}}\right)^h \mbox{ (by definition of } V_j)\\
 & = & \left(\cos(u/\sqrt{N}) \right)^h. 
\end{eqnarray*}
Consequently
  \begin{eqnarray*}
E [e^{i u \alpha_k S_n ((j_k - j_{k-1})/N)}] = [\cos(u\, \alpha_k/\sqrt{N})]^{j_k - j_{k-1}}
. \end{eqnarray*}
  Therefore
 \begin{eqnarray} \label{eq:interst}
 E\left(\left\vert \int_0^1 e^{i u S_n(t)} d \theta_n(t)\right\vert ^{2q}\right) &  = & (q!)^2 \sum_{\epsilon \in T} \,\,\sum_{1\leq j_1\leq  \ldots\leq j_{2q} \leq 2^n} c(j_1) c(j_2) \ldots c(j_{2q}) \times  \nonumber\\
&&  \prod_{k=1}^{2q} [\cos(u\, \alpha_k /\sqrt{N})]^{j_k - j_{k-1}}.
 \end{eqnarray}
Since $\alpha_j$ is the sum of $2q-j+1$ numbers each equal to $\pm 1$, it follows that  $\alpha_j \ne 0$ for all even $j$. Clearly, for any $\epsilon = (\epsilon_1, \ldots, \epsilon_{2q}) \in T$, we have that $|\alpha_k| \leq q$ for all $ 1\leq k \leq 2q$. 

From the condition imposed on the numbers $n$, $q$, and $u$ ($q \leq \log n$, $n \leq u \leq n+1$, $N = 2^n$) it is clear that 
$u/\sqrt{N}$ is small enough to ensure  that  $0< \cos(u\, h/\sqrt{N})
\leq \cos(u/\sqrt{N})$ for any $1 < h \leq q$. In
particular $$\cos\left(u\, \alpha_k /\sqrt{N}\right) \leq \cos\left(u/\sqrt{N}\right) \mbox{ for all even } k.$$  For $k$ odd, we use the obvious inequality  $\cos(u\, \alpha_k /\sqrt{N}) \leq 1$. 
 Consequently
  \begin{eqnarray*}
 E\left(\left\vert \int_0^1 e^{i n S_n(t)} d \theta_n(t)\right\vert ^{2q}\right) & \leq & (q!)^2 \sum_{\epsilon \in T} \,\,\sum_{1\leq j_1\leq  \ldots\leq j_{2q} \leq 2^n} c(j_1) c(j_2) \ldots c(j_{2q}) \times \\
&&  [\cos(u/\sqrt{N})]^{j_2 - j_{1}} [\cos(u/\sqrt{N})]^{j_4 - j_{3}} \times\\
 &&\times \ldots [\cos(u/\sqrt{N})]^{j_{2q} - j_{2q-1}}. 
 \end{eqnarray*}
 Since the cardinality of $T$ is $(2q)!/(q!^2)$,  and all the terms are positive,  
 \begin{eqnarray*}
 E\left(\left\vert \int_0^1 e^{i n X_n(t)} d \theta_n(t)\right\vert ^{2q}\right) & \leq & 
 (2q)! \sum_{j_1 = 1}^{N} \sum_{j_2 = j_1}^{N} \sum_{j_3 = j_2}^{N} \ldots \sum_{j_{2q} = j_{2q-1}}^{N}  c(j_1) c(j_2) \ldots c(j_{2q}) \times \\
&& [\cos(u/\sqrt{N})]^{j_2 - j_{1}} [\cos(u/\sqrt{N})]^{j_4 - j_{3}} \times\\
 &&\ldots [\cos(u/\sqrt{N})]^{j_{2q} - j_{2q-1}}\\
 & \leq & (2q)! \sum_{1\leq j_1\leq j_3 \leq \ldots \leq j_{2q-1} \leq N} \sum_{j_2 = j_1}^{j_3} \sum_{j_4 = j_3}^{j_5} \ldots  \sum_{j_{2q} = j_{2q-1}}^{N}\\
&& c(j_1) c(j_2) \ldots c(j_{2q}) \times \\
&& [\cos(u/\sqrt{N})]^{j_2 - j_{1}} [\cos(u/\sqrt{N})]^{j_4 - j_{3}} \times\\
 &&\ldots \times [\cos(u/\sqrt{N})]^{j_{2q} - j_{2q-1}}.
 \end{eqnarray*}
By setting
 $h_k = j_{k} - j_{k-1}$ for $k$ even, and by taking $c(t) = 0$ for $t > 2^n$, we obtain 
 \begin{eqnarray*}
 E\left(\left\vert \int_0^1 e^{i u S_n(t)} d \theta_n(t)\right\vert ^{2q}\right)  &\leq & (2q)! 
\sum_{1\leq j_1\leq j_3 \leq \ldots \leq j_{2q-1} \leq N} \sum_{h_2 = 0}^{N} \sum_{j_4 = 0}^{N} \ldots  \sum_{h_{2q} = 0}^{N}\\
&& c(j_1) c(j_3) \ldots c(j_{2q-1}) \times \\
&& c(j_1+ h_2) c(j_3 +h_4) \ldots c(j_{2q-1} + h_{2q})\\
&& [\cos(u/\sqrt{N})]^{h_2} [\cos(u/\sqrt{N})]^{h_4} \times\\
 &&\ldots \times [\cos(u/\sqrt{N})]^{h_{2q}}.
 \end{eqnarray*} 
We now estimate the sum 
 $$M = \sum_{h=0}^{N} c(r + h) [\cos(u/\sqrt{N})]^{h},\,\ r \leq N. $$ Clearly, 
       $$M = \int_0^1\left[\cos\left(u/\sqrt{N}\right)\right]^{N t} d\mu(t)\,\,\mbox{ where } \mu = \sum_{j=1}^{N} c(j+r) \delta_{j/N} $$ and $c(h) = 0$ for $h \notin \{1, 2, \ldots, N\}.$
       
By  Lemma \ref{dakjasfad}, we obtain that 
 $$\sum_{h=0}^{N} c(r + h)\left[\cos\left(u/\sqrt{N}\right)\right]^{h}  = \mu[0, 1]\, a^{N}  - N \log a \int_0^1 a^{t\, N} \mu[0, t]dt $$
where $$ a = \cos\left(u/\sqrt{N}\right).$$ 
Since $\theta$ is an $\alpha$-Frostman measure and $\theta_n$ is its $n$-step approximation, it follows from relation (\ref{eadagasdw}) that, \begin{eqnarray}
\mu[0, t] = \theta_n[r/N, t+ r/N]\leq C t^{\alpha}\,\, \mbox{ for } t \geq 1/N
\end{eqnarray} 
and \begin{eqnarray} \label{wewedfw}
\theta_n(I) \leq 1/N^{\alpha}
\end{eqnarray} for every interval $I$ such that $\vert I \vert \leq
1/N$.
We also have that $0 < a < 1$. 
Then, 
  \begin{eqnarray*}
\sum_{h=0}^{N} c(r + h)\left[\cos\left(u/\sqrt{N}\right)\right]^{h} & \leq & \theta_n[0, 1]a^{N} \\
 && -N \log a \left(\int_0^{1/N} a^{t\, N} \theta_n[r, r+t] dt  + C \int_{1/N}^1 a^{t\, N} t^\alpha dt \right).
\end{eqnarray*}
By (\ref{wewedfw}), we have that \begin{eqnarray*}
N \log a\int_0^{1/N} a^{t\, N} \theta_n[r, r+t] dt & \leq  & N\log a \int_0^{1/N} a^{t\, N} (1/N^{\alpha})\, dt \\
& = & (a-1)/N^{\alpha}
\end{eqnarray*}
 and 
$$C N \log a \int_{1/N}^1 a^{t\, N} t^\alpha dt  = \frac{C}{(-N \log a)^\alpha} \Gamma(\alpha + 1). $$
Therefore (using $\theta_n[0, 1] = \theta[0,1] \leq 1$),
\begin{eqnarray} \label{32wqwrwsre}
\sum_{h=0}^{N} c(r + h)\left[\cos\left(u/\sqrt{N}\right)\right]^{h} & \leq &  a^{N} + \frac{1 - a}{N^{\alpha}} + \frac{C}{(N \log(1/a))^\alpha} \Gamma(\alpha + 1).
\end{eqnarray}
We next estimate each term in  (\ref{32wqwrwsre}). 
\begin{enumerate}[label=\arabic*.]
\item  Clearly,  $(1-a)/N^{\alpha} \leq 1/u^{4\alpha}$.

\item For the quantity $a^{N}$ with  $a = \cos\left(u/\sqrt{N}\right)$, 
 we use the following estimates  (``Taylor expension"):
 \begin{eqnarray*}
 \cos(x) & \leq &  1-x^2/2 + x^4/4!\\
 \log(1-x) & \leq &  -x  \mbox{ for } x \geq 0\\
 \log(\cos(x)) & \leq &  -\frac{x^2}{2} + \frac{x^4}{4!} \\
 \end{eqnarray*}
 Therefore, 
  \begin{eqnarray*}
  \log a^{N} & = & N \log \cos\left(u/\sqrt{N}\right)\\
& \leq & N \left(-\frac{u^2}{2\times N } + \frac{u^4}{4!\times N^{2}}\right)\\
& = & -\frac{u^2}{2} + \frac{u^4}{4! \times N}. 
  \end{eqnarray*}
Now since (for large $n$),  $u$ is such that $u^2 \leq \sqrt{N}$, it follows that
\[\log a^{N}  \leq -\frac{u^2}{2} + 1 \leq -\frac{u^2}{3}\] for  large $u$. Therefore,
 $a^{N} \leq e^{-\frac{u^2}{3}}.$ This quantity converges to zero for $u \rightarrow \infty$. 
   
\item For the quantity $$\frac{C}{(N \log(1/a))^\alpha},$$
    we have that 
      \begin{eqnarray*}
      N \log(1/a) & =&  - N \log \cos\left(u/\sqrt{N}\right)\\
& \geq &  -N (-\frac{u^2}{2\times N} + \frac{u^4}{4!\times N^{2}}) \\
& \geq &  \frac{u^2}{2} - 1  \\
& \geq & \frac{u^2}{3}\, \mbox{ (for large } n ),  
      \end{eqnarray*}
Therefore $$\frac{C}{(N \log(1/a))^\alpha} \leq 3 C/u^{2\alpha}$$
and 
  \begin{eqnarray*}
\sum_{h=0}^{N} c(r + h) \left[\cos\left(u/\sqrt{N}\right)\right]^{h} & \leq & e^{-u^2/3} + \frac{1}{u^{4\alpha}} + \frac{3C\Gamma(\alpha +1)}{u^{2\alpha}}\\
&\leq & \frac{11}{u^{2\alpha}}
\end{eqnarray*}
 by taking $C = 3$ and $\alpha \in [0, 1]$ and using the fact that the  quantities  $e^{-u^2/3}$ and $1/u^{4\alpha}$  are dominated by $1/u^{2\alpha}$ for relatively large $u$. 
 Hence
\begin{eqnarray*}
  E\left(\left\vert \int_0^1 e^{i u S_n(t)} d \theta_n(t)\right\vert ^{2q}\right)  &\leq & (2q)! (11\, u^{-2\alpha})^q \times\\
&&\sum_{1\leq j_1\leq j_3 \leq \ldots \leq j_{2q-1} \leq N} c(j_1) c(j_3) \ldots c(j_{2q-1}) .
\end{eqnarray*}  
Finally, 
   \begin{eqnarray*}
 \sum_{1\leq j_1\leq j_3 \leq \ldots \leq j_{2q-1} \leq 2^n} c(j_1) c(j_3) \ldots c(j_{2q-1}) & = & \int_{0\leq t_1 \leq t_2 \ldots t_{q}\leq 1} d\theta(t_1) d\theta(t_2) \ldots d\theta(t_q),\\
 \end{eqnarray*}
 and we observe that, the integral $$\int_{0\leq t_{\sigma(1)} \leq t_{\sigma(2)} \ldots t_{\sigma(q)}\leq 1} d\theta(t_1) d\theta(t_2) \ldots d\theta(t_q)$$ is the same  for all permutations $\sigma$ of $\{1, 2, \ldots, q\}$.
 Therefore,   
      \begin{eqnarray*}
\int_{0\leq t_1 \leq t_2 \ldots t_{q}\leq 1} d\theta(t_1) d\theta(t_2) \ldots d\theta(t_q) & = &\frac{1}{q!} \int_0^1 \int_0^1 \ldots\int_0^1 d\theta(t_1) d\theta(t_2) \ldots d\theta(t_q) \\
& \leq & \frac{1}{q!}
\end{eqnarray*} 
since $\theta[0, 1] \leq 1$. 

It follows that 
 \[
  E\left(\left\vert \int_0^1 e^{i u S_n(t)} d \theta_n(t)\right\vert
  ^{2q}\right)\ \leq\ \frac{(2q)!}{q!} (11\, u^{-2\alpha})^q\ \leq\ (22\,q\,u^{-2\alpha})^q.\eqno{\qEd}
\]
\end{enumerate}

\end{document}